%% file: robust-population-protocols.tex
\documentclass[a4paper,UKenglish,cleveref, autoref, authorcolumns, thm-restate]{lipics-v2021}
\hideLIPIcs  %uncomment to remove references to LIPIcs series (logo, DOI, ...), e.g. when preparing a pre-final version to be uploaded to arXiv or another public repository

\title{Monadic Presburger Predicates have Robust Population Protocols}

\author{Philipp Czerner}{Technical University of Munich, Germany}{philipp.czerner@in.tum.de}{https://orcid.org/0000-0002-1786-9592}{}

\author{Javier Esparza}{Technical University of Munich, Germany}{javier.esparza@itum.de}{https://orcid.org/0000-0001-9862-4919}{}

\author{Vincent Fischer}{Technical University of Munich, Germany}{vincent.fischer@tum.de}{https://orcid.org/0009-0009-3071-0736}{Funded by the Research Training Group 2428 ``ConVeY'' of the German Research Council.}

\author{Roland Guttenberg}{Technical University of Munich, Germany}{roland.guttenberg@tum.de}{https://orcid.org/0000-0001-6140-6707}{}

\author{Julian Pins}{Technical University of Munich, Germany}{Julian.Pins@tum.de}{}{}

\author{Simon Reilich}{Technical University of Munich, Germany}{simonreilich@t-online.de}{}{}

\authorrunning{P. Czerner, J. Esparza, V. Fischer, R. Guttenberg, J. Pins, and S. Reilich}

\Copyright{Philipp Czerner, Javier Esparza, Vincent Fischer, Roland Guttenberg, Julian Pins, and Simon Reilich}



\category{} %optional, e.g. invited paper

\relatedversion{} %optional, e.g. full version hosted on arXiv, HAL, or other respository/website
\funding{This work was partially supported by project 551606821 ``Verification and Synthesis of Dynamic Network Algorithms'' of the German Research Council.}

\nolinenumbers %uncomment to disable line numbering

\usepackage{mathtools}
\usepackage{amsthm}
\usepackage{bm}
\usepackage[bb=boondox]{mathalfa}
\usepackage{stmaryrd}

\input{macros.tex}

\begin{document}

\maketitle

\begin{abstract}
Population protocols are a model of distributed computation in which a collection of indistinguishable finite-state agents interact randomly in pairs to decide a predicate of their initial configuration. The agents decide by achieving a stable consensus on whether the predicate holds or not. It is known that population protocols can decide exactly the predicates expressible in Presburger arithmetic.

Recently, Lossin et al. have introduced a notion of protocol robustness against adversarial crash failures. They show that all \emph{atomic} Presburger predicates can be decided by robust protocols, and ask whether the same holds for \emph{every} Presburger predicate. We make progress towards settling this question by proving that all predicates expressible in monadic Presburger arithmetic have robust protocols. In addition, we analyze the cost of robustness in terms of state complexity. We study the ratio between the number of states of the smallest robust protocol for a given predicate and the smallest protocol for it. We show that the cost of robustness is at least double exponential in the size of the predicate, and prove that the robust protocols by Lossin \etal for  threshold predicates $x \geq k$ have optimal state complexity.
\end{abstract}

\input{sec-intro.tex}
\input{sec-prelims.tex}
\input{sec-strongly-robust-protocols.tex}

%\input{sec-optimal-tower.tex}
\input{sec-lower-bound.tex}
\input{sec-conclusion.tex}

\bibliography{robust-population-protocols}

\end{document}

%% file: macros.tex
\newcommand{\etal}{\textit{{et al. }}}

\newcommand{\N}{\mathbb{N}}
\newcommand{\Z}{\mathbb{Z}}
\newcommand{\Prot}{P}
\newcommand{\Config}[1]{#1}
\newcommand{\Cfg}{\Config{C}}
\newcommand{\Conf}[1]{\Config{{#1}}}
\newcommand{\Iconf}[1]{\Config{{#1}}}
\newcommand{\ICfg}{\Config{A}}

\def \ifemptycheck#1{\def\temp{#1} \ifx\temp\empty }
\newcommand{\emptymultiset}{\emptyset}
\def \multiset#1{ \ifemptycheck{#1} \emptymultiset \else \Lbag#1\Rbag \fi }
\newcommand{\Abs}[1]{\mathopen|#1\mathclose|}
\newcommand{\Support}[1]{\left\llbracket#1\right\rrbracket}
\newcommand{\pre}{\textit{pre}}
\newcommand{\post}{\textit{post}}

\newcommand{\step}{\rightarrow}
\newcommand{\steps}{\xrightarrow{\ast}}
\newcommand{\snipe}{\hookrightarrow}
\newcommand{\snipereach}[1]{\xhookrightarrow{\ast}{{#1}}}

\DeclarePairedDelimiterX{\set}[1]{\lbrace}{\rbrace}{#1}
\DeclarePairedDelimiterX{\setc}[2]{\lbrace}{\rbrace}{#1 \;\delimsize\vert\; #2}
\DeclarePairedDelimiter{\supp}{\llbracket}{\rrbracket}

\DeclareMathOperator*{\pl}{pl}
\DeclareMathOperator*{\lcm}{lcm}

\newcommand{\Copies}{c}
\newcommand{\Height}{T}
\newcommand{\Values}{V}
\newcommand{\VValues}{\overrightarrow{\Values}}
\newcommand{\values}{v}
\newcommand{\vvalues}{\vec{\values}}
\newcommand{\Outputs}{\Omega}
\newcommand{\VOutputs}{\overrightarrow{\Omega}}
\newcommand{\outputs}{\omega}
\newcommand{\voutputs}{\vec{\outputs}}

\newcommand{\Trans}[1]{\tag*{$\langle$\textsc{#1}$\rangle$}}
\newcommand{\reconcile}{\sqcap}

\makeatletter
\newlength{\proto@lw}
\newenvironment{protocol}[1]{%
    \par\smallskip%
    \noindent\textbf{Protocol }\textsc{#1}\textbf{.}\nopagebreak%
    \protected@edef\@currentlabel{\textsc{#1}}%
    \begin{list}{}{%
        \setlength{\leftmargin}{0pt}%
        \setlength{\labelwidth}{\proto@lw}%
        \setlength{\labelsep}{0pt}%
        \setlength{\itemindent}{\proto@lw}%
        \setlength{\itemsep}{2pt}%
        \setlength{\parsep}{0pt}%
        \setlength{\topsep}{\smallskipamount}%
    }%
    \setlength{\abovedisplayskip}{4pt plus 1pt minus 1pt}%
    \setlength{\belowdisplayskip}{4pt plus 1pt minus 1pt}%
    \setlength{\abovedisplayshortskip}{2pt plus 1pt}%
    \setlength{\belowdisplayshortskip}{2pt plus 1pt}%
    \newcommand{\HeadSingle}[1]{%
        \item[{\makebox[\proto@lw][l]{\textbf{\Title.}}}]##1%
    }%
    \newcommand{\Head}{\HeadSingle}%
    \newcommand{\States}{\def\Title{States}\Head}%
    \newcommand{\Input}{\def\Title{Input}\Head}%
    \newcommand{\Output}{\def\Title{Output}\Head}%
    \newcommand{\Transitions}{\def\Title{Transitions}\Head}%
}{%
    \end{list}%
}

\makeatother

%% file: sec-intro.tex
\section{Introduction}
Population protocols are a much studied model of distributed computation in which a collection of indistinguishable finite-state agents interact randomly in pairs with the goal of deciding a property of their initial configuration \cite{AngluinADFP06,AspnesR09,AlistarhG18,ElsasserR18}. The decision is taken by stable consensus: Eventually all agents agree on whether the property holds or not, and never change their mind again.  We say that the protocol (stably) \emph{converges} to ``yes'' or ``no''.

Since agents are indistinguishable, a configuration of a protocol can be modeled as a mapping assigning to each state $q$ the number of agents currently in $q$.  
Given a protocol with initial states $q_1, \ldots, q_n$, a property of the protocol can be modeled as a predicate $P(x_1, \ldots, x_n) \colon \mathbb{N}^n \to \{0,1\}$: The initial configurations satisfying the property are those mapped to $1$ by $P$. It was shown in \cite{AngluinAER07} that population protocols can decide exactly the Presburger predicates, that is, the predicates expressible in the theory $\textit{Th}(\mathbb{N}, +, <, 0, 1)$.

Population protocols have been used to model ad hoc networks of tiny sensors, and chemical reaction networks, in which agents are molecules and interactions are chemical reactions.  In both settings the number of agents can be very large, and so failures are bound to occur. This raises the question of which fragment of the Presburger predicates can be decided \emph{robustly}, even in the presence of failures. In this paper we consider \emph{adversarial crash failures} (called just \emph{crashes} or \emph{failures} in the following) that cause an agent to stop interacting with other agents. Adversarial means that the adversary can choose which agent crashes and at which moment during the execution of the protocol.

Robust decision under crash failures was first studied by Delporte-Gallet \etal \cite{Delporte-GalletFGR06}. Loosely speaking, they showed that all Presburger predicates can be decided robustly under the assumption that the number of failures is bounded by a constant, independent of the number of agents. 
%More precisely, given an upper bound $b$ on the number of failures, the paper shows how to transform a protocol deciding a property into another one that decides the property under the assumption that at most $b$ agents fail. This approach had the problem that the number of states of the new protocol grows exponentially in $b$. making the construction very difficult to implement when the number of agents--- and so the number of failures---is very large, for example when agents model molecules.  
In \cite{LossinCEGP25}, Lossin \etal study if this assumption can be lifted. They propose a stronger notion of robustness. To give the intuition, consider the Presburger predicate $x \geq k$, corresponding to the property ``the initial configuration contains at least $k$ agents''.  A protocol deciding $x \geq k$ is robust if for every configuration with $k'>k$ agents the following holds: if at most $k'-k$ agents fail, the protocol converges to ``yes''. Observe that the number of failures a robust protocol must tolerate is no longer constant. Further, it is the best possible number. Indeed, assume that more than $k'-k$ agents fail. If the adversary makes them fail at the start of the protocol, before they interact with any other agent, the ``survivors'' do not even know that these agents exist, and so they cannot possibly converge to ``yes''.  

Lossin \etal generalize this definition of robustness to arbitrary predicates. A protocol robustly decides a predicate $P(x_1, \ldots, x_n)$ if for every initial configuration $a=(a_1, \ldots, a_n) \in \mathbb{N}^n$ and every initial configuration $b \leq a$ obtained by letting $a-b$ agents crash before the protocol starts, the following holds: if $P(a)=P(b)$, then the protocol converges to $P(a)$ even if the crash times of the agents of $a-b$ are chosen by the adversary. They ask the question whether every Presburger predicate can be robustly decided. They exhibit robust protocols for \emph{threshold} and \emph{modulo} predicates of the form $\sum_{i=1} a_i x_i \leq b$ for some $a_1, \ldots, a_n, b \in \mathbb{Z}$, and $\sum_{i=1} a_i x_i \leq b \pmod{m}$ for some $a_1, \ldots, a_n \in \mathbb{Z}$ and $b, m \in \mathbb{N}$ with $b < m$, respectively. However,
whether every Presburger predicate can be robustly decided remains open.

%Since every Presburger predicate is equivalent to a Boolean combination of threshold and modulo predicates (see e.g. \cite{Haase18}), every Presburger predicate can be robustly decided if{}f the robustly decidable predicates are closed under Boolean operations. However, these questions are open.

Our first contribution shows that every \emph{monadic} Presburger predicate can be robustly decided.  A Presburger predicate is monadic if it is represented by a formula in which each threshold and modulo atomic formula contains at most one variable\footnote{Notice that monadic formulas may have multiple variables; for example, $(x \leq 3 \vee y \geq 4) \wedge z \bmod 7 \leq 2$ is monadic.}. To our knowledge, this is the first fragment of Presburger arithmetic closed under logical operations for which robust decidability is proved.  

In our second contribution we study the cost of robustness in terms of \emph{state complexity}, defined as the number of states of the smallest protocol deciding  a predicate. We obtain a lower bound for the state complexity of \emph{asymptotically constant predicates} defined as the predicates $P(x_1, \ldots, x_n)$ for which there exists a $k$ such that $P(x_1, \ldots, x_n)$ returns the same value for all inputs $(x_1, \ldots, x_n)$ satisfying $x_i \geq k$ for every $1 \leq i \leq n$.  As a corollary, we show that the smallest robust protocol for the predicate $x \geq k$ has exactly $k$ states. It is known that $x \geq t$ is decidable by protocols with $\Theta(\log \log t)$ states for infinitely many $t \geq 0$, and every protocol needs $\Omega(\log \log t)$ \cite{BlondinEJ18,CzernerGHE24,CzernerEL23,Leroux22}. Therefore, the price of robustness in terms of state complexity is at least double exponential. 

\smallskip\noindent\textbf{Related work.} We have already discussed the work of \cite{Delporte-GalletFGR06} and \cite{LossinCEGP25} on crash failures.  Our paper extends this line of work. There also exists work on other failure models. In \cite{GuerraouiR09}, Guerraoui and Ruppert extend population protocols by assigning unique identities to agents and construct protocols for all  predicates in NSPACE($\log n)$ that are  robust with respect to a constant number of Byzantine failures. Alistarh \emph{et al.} have studied robustness of population protocols and chemical reaction networks with respect to failures in which agents can silently and probabilistically move to adversarially chosen states \cite{AlistarhDKSU17,Alistarh0U21}.  Di Luna \emph{et al.} study in \cite{LunaFIISV19,LunaFIISV20}  omission faults in which an agent fails to read its partner's state during an interaction.

\smallskip\noindent\textbf{Structure of the paper.} Section
\ref{sec:prelims} introduces general notation and formally defines population
protocols. \Cref{sec:robustness} defines robustness of protocols and Section \ref{sec:protocols} presents robust protocols for monadic predicates. Section 
\ref{sec:lowerbound} then proves a lower bound on the state complexity of robust protocols. Section \ref{sec:conclusions} contains conclusions.

%% file: sec-prelims.tex
\section{Preliminaries}
\label{sec:prelims}
\noindent\textbf{Natural Numbers.} We denote the set of all natural
numbers as \(\N \coloneq \{0,1,2,\ldots\}\) and the set of all natural
numbers modulo \(m\) as \(\Z_{m} \coloneq \{0,1,\ldots,
m-1\}\). Additionally we define the set \([k]\coloneq\{1,2,\ldots,k\}\).

\medskip\noindent\textbf{Multisets.} Let $E$ be a finite set. A multiset over $E$ is a mapping $E \rightarrow \N$, and $\N^E$ denotes the set of all multisets over $E$. We sometimes write multisets using set-like notation, e.g.\ $\multiset{a, 2 \cdot b}$ denotes the multiset $v$ such that $v(a) = 1$, $v(b) = 2$ and $v(e) = 0$ for every $e \in E \setminus \set{a,b}$.
The empty multiset $\multiset{ }$ is also denoted $\multiset{}$.  For $E' \subseteq E$, $v(E') := \sum_{e\in E'} v(e)$ is the number of elements in $v$ that are in $E'$.  The \emph{size} of $v \in \N^E$ is $\Abs{v} := v(E)$. The \emph{support} of $v \in \N^E$ is the set $\Support{v} := \{ e \in E \mid v(e) > 0 \}$. 
%If $E \subseteq \Z$, then we let $\multisum{v}:= \sum_{e \in E} e \cdot v(e)$ denote the sum of all the elements of $v$. 
Given $u,v \in \N^E$, we let $u \leq v$ denote that $u(e) \leq v(e)$ holds for every $e \in E$; further, we let $u + v$ and $u-v$ denote the multisets given by $(u+v)(e):=u(e)+v(e)$ 
and $(u-v)(e):=u(e)-v(e)$ for every $e \in E$. The latter is only defined if $u \geq v$.

\medskip\noindent\textbf{Population Protocols.}
  Let $K$ be a finite set of \emph{outputs}. A \emph{population protocol} over $K$ is a four-tuple \(\Prot = (Q, I,
  O, \delta)\), where \(Q\) is a finite set of \emph{states},  \(I \subseteq Q\) is the set of \emph{initial states},  \(O : Q \to K\) is the \emph{output function}, and \(\delta : \N^Q_2 \to \N^Q_2\) is the \emph{transition function}, where $\N^Q_2$ denotes the multisets over $Q$ of size exactly 2. 
%   \begin{itemize}
%    \item \(Q\) is a finite set of \emph{states},
%    \item \(I \subseteq Q\) is the set of \emph{initial states},
%    \item \(O : Q \to K\) is the \emph{output mapping}, and
%    \item \(\delta : \N^Q_2 \to \N^Q_2\) is the \emph{transition function}, where $\N^Q_2$ denotes the multisets over $Q$ of size exactly 2. 
%  \end{itemize}

If $\delta(\multiset{q_{1},q_{2}}) = \multiset{q_{1}',q_{2}'}$,  then we write  $q_{1},q_{2}~\mapsto~q_{1}',q_{2}'$ and call this expression a \emph{transition}. Given a transition $t =  \; q_{1},q_{2} \mapsto q_{1}',q_{2}'$, we let $\pre(t) := \multiset{q_1, q_2}$ and $\post(t) := \multiset{q_1', q_2'}$. If $\delta(\multiset{q_{1},q_{2}})$ is not explicitly defined, then we assume $\delta(\multiset{q_{1},q_{2}})= \multiset{q_{1},q_{2}}$.

\medskip\noindent\textbf{Configurations, steps, (fair) executions.}
 A \emph{configuration} of \(\Prot=(Q,I,O,\delta)\) is a multiset $\Cfg \in \N^{Q}$.  \(\Cfg(q)\) models the number of agents in state \(q\), and \(|C| := \sum_{q\in Q} C(q) \) denotes the total number of agents. $\Cfg$ \emph{populates} state $q$ if $\Cfg(q) > 0$.  The \emph{support} of $\Cfg$, denoted $\supp \Cfg$, is the set of states populated by $\Cfg$. \(\Cfg\) is \emph{initial} or an \emph{input} if \(\supp \Cfg \subseteq I\), that is, if it only populates initial states. 
\(\Cfg\) can make a \emph{step} to \(\Config{D}\), denoted
\(\Cfg \to \Config{D}\), if $\Cfg = \Config{D}$ or there exists a transition $t$ such that $\Cfg \geq \pre(t)$ and $\Config{D}= \Cfg - \pre(t) + \post(t)$.
We let \(\steps\) denote the reflexive transitive closure of \(\step\). If \(\Cfg \steps \Config{D}\), then we say that \(\Config{D}\) is \emph{reachable} from \(\Cfg\).

An \emph{execution} is an infinite sequence of configurations
  \(\pi = \Cfg_{0}, \Cfg_{1}, \Cfg_{2},\ldots\) such that \(\Cfg_{i} \step \Cfg_{i+1}\) for all \(i
  \in \N\).  An execution \(\pi\) is  \emph{fair} if for every two configurations
  \(\Cfg, \Config{D}\), if \(\Cfg\) occurs infinitely often in \(\pi\)
  and \(\Cfg \steps \Config{D}\) then \(\Config{D}\) also occurs
  infinitely often in \(\pi\). 

%Following  \cite{LossinCEGP25} we introduce the notion of a \emph{snipe} to model the failure of an agent by completely stopping to interact.
%
%\begin{definition}[Snipes]
%  If \(\Config, \ConfigB\) are two configurations, such that
%  \(\ConfigB \subseteq \Config\) and \(\left| \Config - \ConfigB \right| =
%  1\), then we call \(\Config \snipe \ConfigB\) a \emph{snipe}.
%\end{definition}
%
%\begin{definition}[Executions]
%  An execution with at most $k$ snipes is an infinitely long sequence of configurations,
%  \(\pi = \Config_{0}, Config_{1}\Config_{2},\ldots\), such that \(\Config_{i} \step \Config_{i+1}\) or \(\Config_{i} \snipe \Config_{i+1}\) for all \(i
%  \in \N\) and the number of indices \(i\) where \(\Config_{i}
%  \snipe \Config_{i+1}\) is at most \(k\).
% % If the number of snipes is not explicitly stated, it is assumed to be 0.
%
%  An execution \(\pi\) is called \emph{fair} if for every two configurations
%  \(\Config, \ConfigB\), if \(\Config\) occurs infinitely often in \(\pi\)  
%  and \(\Config \steps \ConfigB\) then \(\ConfigB\) also occurs
%  infinitely often in \(\pi\). 
%\end{definition}

\medskip\noindent\textbf{Consensus, stable consensus, computed function.} Let $r \in K$ be an output value. A configuration \(\Cfg\) is an \(r\)-\emph{consensus} if \(O(\supp \Cfg) = \{r\}\), that is, all agents
have output $r$. Further, $\Cfg$ is a \emph{stable} \(r\)-consensus if every configuration reachable from $\Cfg$ is also a \(r\)-consensus. The protocol \(\Prot\) computes a function \(f: \N^I \to K\) if for every input \(\Cfg_{0} \in \mathbb{N}^I\) every fair execution starting at \(\Cfg_{0} \in \mathbb{N}^I\) \emph{converges} to \(f(\Cfg_{0})\), meaning that it eventually reaches a stable \(f(\Cfg_{0})\)-consensus.
If $K=\{0,1\}$, then we also call $f$ a \emph{predicate} and say that the protocol \emph{decides} $f$. Notice that not every protocol computes a function. Observe also that the arity of the function computed by a protocol, if any, is equal to its number of initial states.

\begin{example}\label{ex:pebbles} We introduce \ref{prot:pebbles}, a simple protocol deciding the
  threshold predicate \(x \ge k\) for some fixed \(k \in \N\). Intuitively, each agent initially carries one pebble,  and some agent eventually collects all pebbles from all agents. If the number of pebbles is at least \(k\),
  then it tells other agents that the threshold has
  been reached. 
\begin{protocol}{Pebble}
\label{prot:pebbles}
\States{$Q = \set{0, 1, \ldots, k-1, k}$. (Each agent may hold up to \(k\) pebbles.)}
\Input{ $I = \set{1}$. (Each agent starts with one pebble.)}
\Transitions{For every $i, j \in Q$: 
\begin{itemize}
\item $i,j \mapsto i + j, 0$ \, if \(i + j < k\).  (One agent gives its pebbles to the other.)
\item $i,j \mapsto k, k$ \;\;\;\;\;\, if \(i + j \geq  k\). (State $k$ models ``I know the threshold is reached''.)
\end{itemize}}
\Output{$O(q) =$ \textbf{if} $q=k$ \textbf{then} $1$ \textbf{else} $0$. (Agents in state $k$ output 1, others 0.)}
\end{protocol}
%\noindent Notice that in some executions no agent ever collects all pebbles, but these executions are not fair.
\end{example}

\noindent\textbf{Computational power of population protocols.} Presburger arithmetic is the first-order theory of addition, usually defined as the theory $\textrm{Th}(\mathbb{N}, +, < , 0, 1)$. We consider
here the equivalent theory whose atomic formulas are 
\emph{threshold} formulas of the form $\sum_{i=1}^n a_i x_i \leq b$ for some $a_1, \ldots, a_n, b \in \mathbb{Z}$ and modulo formulas of the form $(\sum_{i=1}^n a_i x_i) \bmod m \leq b$ for some $a_1, \ldots, a_n \in \mathbb{Z}$ and $b, m \in \mathbb{N}$ with $b < m$. The theory is obtained by closing the atomic formulas under Boolean operations and existential quantification. The semantics of a formula $F(x_1, \ldots, x_n)$ with free variables $X=\{x_1, \ldots, x_n\}$ is the predicate $p_F : \mathbb{N}^X \to \{0,1\}$ defined in the expected way, e.g. for $F(x,y) = 2x - 3y > 5$ we have $p_F(n_1, n_2) = 1$ if{}f $2n_1 - 3n_2 > 5$; the predicates of Presburger formulas are called Presburger predicates. This version of Presburger arithmetic has a quantifier elimination procedure, and so every Presburger formula is equivalent to a Boolean combination of threshold and modulo formulas (resp. predicates) \cite{Haase18}. 

Let $K$ be a finite set. A function $f \colon \mathbb{N}^X \to K$ is a \emph{Presburger function} if 
for every $i \in K$ there is a Presburger formula $F_i(x_1, \ldots, x_n)$ such that  $f(a_1, \ldots, a_n) = i$ if{}f $p_{F_i}(a_1, \ldots, a_n)=1$.  Angluin \etal show in \cite{AngluinAER07} that a function  $f \colon \mathbb{N}^I \to K$ is computable by a protocol with $I$ as set of initial states if{}f $f$ is Presburger. In particular, a predicate $f \colon \mathbb{N}^I \to \{0,1\}$ can be computed or \emph{decided} by a population protocol if{}f $f$ is a Presburger predicate. 

%%% Local Variables:
%%% TeX-master: "robust-population-protocols"
%%% End:

\section{Snipes and robustness} \label{sec:robustness} 
The above notions capture the failure-free behavior of a
protocol. In practice, however, agents may crash in an adversarial way. We introduce the failure model of \cite{LossinCEGP25} and a notion of robustness. The notion is slightly stronger than the one of  \cite{LossinCEGP25}. At the end of the section we explain the difference and why a new notion is needed. 

In \cite{LossinCEGP25}, failures are modeled as \emph{snipes}. Intuitively, a snipe deletes an agent from a configuration. 

\begin{definition}[Snipes, executions with snipes]
  Let  \(\Cfg, \Config{D}\) be configurations. If
  \(\Config{D} \leq \Cfg\) and \(\left| \Cfg - \Config{D} \right| =
  1\) we write \(\Cfg \snipe \Config{D}\) and call this expression a \emph{snipe step}, or just a \emph{snipe}.
An \emph{execution with at most $k$ snipes} is an infinite sequence of configurations,
  \(\pi = \Cfg_{0}, \Cfg_{1}, \Cfg_{2},\ldots\) such that \(\Cfg_{i} \step \Cfg_{i+1}\) or \(\Cfg_{i} \snipe \Cfg_{i+1}\) for all \(i
  \in \N\) and the number of indices \(i\) where \(\Cfg_{i}
  \snipe \Cfg_{i+1}\) is at most \(k\).
 % If the number of snipes is not explicitly stated, it is assumed to be 0.
The sequence \(\pi\) is \emph{fair} if for every two configurations
  \(\Cfg, \Config{D}\), if \(\Cfg\) occurs infinitely often in \(\pi\)
  and \(\Cfg \steps \Config{D}\) then \(\Config{D}\) also occurs
  infinitely often in \(\pi\). 
\end{definition}

We introduce a notion of robustness. Intuitively, a protocol computing a function \(f: \N^I \to K\) is robust if for any input $\ICfg$ and any number $j \geq 0$ of snipes taking place at adversarial times, the output of the protocol coincides with the output for some input $\ICfg' \leq \ICfg$ such that $|\ICfg' - \ICfg| \leq j$. For example, let $f(x) := \min(x, 6)$ and take the input $x:= 7$. Then, for $1$ snipe a robust protocol must output $6$, for two snipes $6$ or $5$, for three snipes $6$, $5$, or $4$, etc. This fault-tolerance is optimal. Consider for example the case of two snipes. No protocol can guarantee that the output will be only $6$, because the adversary can snipe two agents \emph{before the agents start to interact at all} and the five remaining agents, ignorant that two agents have crashed, can only reach consensus $5$. Similarly, no protocol can guarantee that the output will be exactly $5$, because the adversary can snipe one agent at the beginning, wait until the remaining agents reach stable consensus $6$, and then snipe a second agent; since the consensus with $6$ agents was stable, after sniping one agent the remaining configuration is still a stable consensus.

\begin{definition}[Robustness]
\label{def:robustness}
 A population protocol computing a function \(f: \N^{I} \to K\) is \emph{robust} if
 for every initial configuration \(\ICfg \in \N^{I}\),  every \(j<\vert\ICfg\vert\), and every fair execution with $j$ snipes starting at $\ICfg$, there exists an output value $r \in K$  such that 
 \begin{itemize}
 \item  the execution eventually reaches an $r$-stable consensus, and
 \item \(f(\ICfg') = r\) for some initial configuration \(\ICfg' \leq \ICfg\) satisfying \(|\ICfg -\ICfg'| \leq j\).
\end{itemize}
We say that $r$ is a \emph{permissible output} for $\ICfg$ and $j$.
\end{definition}

\noindent\textit{Remark.} We limit $j$ to be less than $\vert\ICfg\vert$, since
if all agents are sniped, the output is no longer defined at all.

The following example is taken from \cite{LossinCEGP25}.

\begin{example}
\label{ex:tower}
The~\ref{prot:pebbles} family of protocols from Example~\ref{ex:pebbles} is not
robust. Consider the~\ref{prot:pebbles} protocol for $x \geq 100$, and the initial configuration with $199$ agents in state $1$. If the protocol were robust, then it should always output $1$ for any number of snipes between $0$ and $99$. However, one single snipe may suffice to change the output. Indeed, it is easy to see that the configuration in which two agents have $100$ and $99$ pebbles, respectively, and the rest have $0$ pebbles, is reachable from the initial configuration. If the agent with $100$ pebbles is now sniped, then only $99$ pebbles remain, and the protocol outputs $0$.

However, the  \ref{prot:tower} protocol below robustly computes $x \geq k$. Intuitively, at every moment in time each agent of the protocol is at one of the floors or \emph{levels} of a tower of height \(k\). Each agent wants to have a level just for itself, and so if two agents at the same level interact---and so discover they are not alone---then one of them climbs to the next level.
  \begin{protocol}{Tower}\label{prot:tower}
  \States{ $Q = \set{1, 2, \ldots, k}$. (Each agent is at some level in a tower of height \(k\).) }
  \Input{ $I = \{1\}$. (Every agent starts at level \(1\).)}
  \Output{ $O(i) = $ \textbf{if} $i < k$ \textbf{then} $0$ \textbf{else} $1$.  (Agents at level $k$ output 1, others 0.) }
  \Transitions{For every $i < k$:  
  \begin{itemize}
  \item $(i, i) \mapsto (i, i + 1)$. (If two agents meet at the same level (except the top one), one of  them climbs to the next level.)
  \item $(k, i) \mapsto (k, k)$. (An agent at the top level ``pulls up'' any other agent it meets.)
  \end{itemize}
%    If \(i < k\):
%    \begin{equation*}
%      \label{trans:tower-step}
%      \Trans{Step}
%      
%    \end{equation*}
%    An agent on the top level pulls up any other agent it meets.
%    \begin{equation*}
%      \label{trans:tower-pull}
%      \Trans{Pull}
%      (k, i) \mapsto (k, k)
%    \end{equation*}
  }
\end{protocol}
\noindent To see that \ref{prot:tower} is robust, observe that from any
configuration with at least \(n \le k\) agents, at least one agent will
eventually reach level \(n\) or a higher level. Indeed, if no agent is at level \(n\) or higher, 
by the pigeonhole principle some level has at
least two agents, allowing one of them to climb up.
Thus, if the initial configuration has \(k + s\) agents, then even if
up to \(s\) agents are sniped adversarially, at least one of the
remaining agents eventually reaches level \(k\) and
then pulls up all other remaining agents to level \(k\) as well.
\end{example}

\noindent \textit{Remark.} The definition of robustness of
\cite{LossinCEGP25} is for protocols that compute functions \(f:
\N^{I} \to \{0,1\}\) (i.e., protocols that decide predicates). We have
generalized it to functions \(f: \N^{I} \to K\) for any finite set
$K$. However, the restriction of our definition to the case
$K=\{0,1\}$ is slightly stronger than the definition of
\cite{LossinCEGP25}. To see why, consider the predicate $x \geq 7$ and
the input $7$. The definition of \cite{LossinCEGP25} does not put any
constraint on the fair executions with one snipe starting at $7$: some
fair executions may not reach any consensus at all. On the contrary, Definition \ref{def:robustness} requires
that all fair executions converge to either $0$ or $1$. Since protocols that are robust w.r.t. Definition \ref{def:robustness} are better behaved, and all protocols of \cite{LossinCEGP25} either are or can be easily made robust according to it, we consider Definition \ref{def:robustness}  the correct one\footnote{The authors of \cite{LossinCEGP25}, three of whom are also authors of this paper, agree (private communication).}.

%% file: sec-strongly-robust-protocols.tex
\renewcommand{\mod}{{\, \textrm{mod} \,}}
\section{Robust Protocols for Monadic Presburger Functions}
\label{sec:protocols}
As mentioned in the introduction, the question whether every Presburger predicate can be decided by a robust population protocol was left open in \cite{LossinCEGP25}. In this section we give a positive answer for \emph{monadic} Presburger predicates.

\begin{definition}
A formula of Presburger arithmetic is \emph{monadic} if every atomic formula appearing in it contains exactly one variable. A function $f \colon \mathbb{N}^X \to K$  is \emph{monadic} if it is the semantics of some monadic Presburger formula.
\end{definition}

In other words, a formula is monadic if its atomic formulas are of the
form $x_i \ge k$ or $x_i \bmod m \in M$. An example of a monadic formula that we use as running example is
$$
F(x,y) =  \big(x \geq 3 \, \vee \, y \bmod 7 \geq 2  \big)  \, \wedge \, \big( x
\bmod 5 \geq 1 \, \vee \,  y \bmod  2 \geq 1  \big)   \, \wedge \, \big( x \geq 4  \, \vee \, y \geq 17 \big)
$$
An example of a non-monadic  formula is $x -y \leq 3$. Hague \etal have shown that  deciding if a given Presburger formula has an equivalent monadic formula is co-NP-complete  \cite{HagueLRW20}.

It is well known that every monadic Presburger formula is equivalent to a Boolean combination of monadic threshold and modulo formulas \cite{Haase18}. So it suffices to show that every such combination, like $F(x, y)$ above, has a robust protocol. For this, fix an arbitrary
Boolean combination $F(x_1, \ldots, x_n)$ of monadic threshold predicates
$x_i \geq t$ and monadic modulo predicates $(x_i \bmod m) \in M$.  We construct a robust protocol $P_F$ that decides $F$. We will proceed in three steps:

\medskip \noindent \textbf{The functions $f^1, \ldots, f^n$.} For every variable $x_i$, define 
$t^i, m^i \in \mathbb{N}$ as follows:
\begin{itemize}
\item $t^i$ is the maximum of all $t \in \mathbb{N}$ such that the atomic formula $x \geq t$ appears in $F$; if there are none, then $t^i = 0$.
\item $m^i$ is the least common multiple (lcm) of all $m \in \mathbb{N}$ such
  that the atomic formula $x \bmod m \in M$ appears in $F$ for some $M \subseteq
  \Z_{m}$; if there are none, then $m^i=1$.
\end{itemize}
We define the function $f^i \colon \mathbb{N} \to [0,t^i] \times \Z_{m^i}$ as
$f^i(x_i) = (\min(x_i, t^i), x_i \bmod m^i)$.
\noindent In our running example we have $f^x(x) = (\min(x, 4), x \bmod 5)$ and $f^y(y) = (\min(y, 17), y \bmod 14)$.

\medskip \noindent \textbf{The robust protocols $P^1_F, \ldots, P^n_F$.}
For every variable $x_i$, we construct a robust protocol $P^i_F$
computing $f^i(x_i)$. Observe that for every $k \in \mathbb{N}$ the output of
$f^i(k)$ determines the values of \emph{all} the threshold and modulo
atomic formulas of $F$ on $x_i$. For example we have $f^y(26) = (17,
12)$, and from the output $(17, 12)$ we can deduce that $y \bmod 7 \geq
2$ is true, because $y \equiv 12 \pmod{14}$ implies $y \equiv 5
\pmod{7}$; that $y \bmod 2 \geq 1$ is false, because $y \equiv 12 \pmod{14}$ implies $y \equiv 0 \pmod{2}$; and that $y \geq 17$ is true. (Notice, however, that the output does not determine the value of $y$.) In particular, for every $(k_1, \ldots, k_n) \in \mathbb{N}^n$ the outputs of $f^{1}(k_1), \ldots, f^{n}(k_n)$ determine the truth value of every atomic formula of $F$, and so the truth value of $F(k_1, \ldots, k_n)$. This leads to the next step.

\medskip \noindent \textbf{The robust protocol $P_F$.} 
Let $f(x_1, \ldots, x_n) \colon ([0,t^1] \times \Z_{m^1}) \times \cdots \times ([0,t^n] \times \Z_{m^n})$ be the function  defined by
$f(x_1, \ldots, x_n) := (f^1(x_1), \ldots, f^n(x_n))$.
Given robust protocols $P^1_F, \ldots, P^n_F$ computing $f^1, \ldots f^n$, we construct a robust protocol $P_F$ computing $f$. The output of this protocol determines the value of $F$ and so, by adjusting the output mapping of the protocol, we obtain a robust protocol that decides $F$.

In our example, we have
$f(x, y) = (\min(x, 4), x \mod 5, \min(y, 17), y \mod 14)$, and so for instance
$f(14, 26) = (4, 4, 17, 12)$. From $(4, 4, 17, 12)$ we obtain that all
atomic formulas of $F(14, 26)$ are true except $y \bmod 2 \geq 1$. So $F(14, 26)$ is true.

\medskip In the rest of the section we develop this proof outline.  The protocols $P^i_F$ are constructed in three steps: Sections \ref{sec:min} and \ref{sec:mod} construct robust protocols computing the functions $\min(x_i,k)$ and $(x_i \mod k)$, respectively. Section \ref{sec:min-mod} gives a robust  protocol  for the function $(\min(x_i,k), x_i \mod k)$. Section \ref{sec:monadic-predicates} constructs $P_F$ given $P_F^1, \ldots, P_F^n$.

\subsection{Computing $\min(n,k)$ robustly}
\label{sec:min}
\input{sec-min.tex}
\subsection{Computing \((n \bmod m)\) robustly}
\label{sec:mod}
\input{sec-mod.tex}
\subsection{Computing $(\min(n,k), n \bmod m)$ robustly}
\label{sec:min-mod}
\input{sec-min-mod.tex}
\subsection{Monadic Predicates}
\label{sec:monadic-predicates}
\input{sec-monadic-predicates.tex}

%% file: sec-min.tex
We present a robust protocol computing
\(\min(n,k)\) for a given constant \(k\). Recall that
in the \ref{prot:tower} protocol, which robustly decides if $n \ge k$,
an agent who reaches the highest level of the tower pulls all other agents to the top. We modify
\ref{prot:tower} by adding a second component to the states:
In addition to their own level in the tower, agents now also remember
the highest level they have seen any agent occupy in all their
interactions so far, or,  in other words, the highest level they currently
know has already been reached by some agent.
\smallskip\begin{protocol}{RobustMin}\label{prot:robust-min}
  \States{$Q = [k] \times [k]$. Models: (own-level, highest-level-I-know-has-been-reached).}
  \Input{$I = \{(1, 1)\}$}
  \Output{$O(i, h) = h$}
  \Transitions{ When agents with the same level meet, one of them climbs to the next  level.
    Additionally, agents exchange and update their knowledge about the highest occupied level.
    \begin{align*}
      \label{trans:min-step}
      \hspace{-0.4cm}\Trans{Step}
      (i, h_1), (i, h_2) & \mapsto (i, \max\{h_{1},h_{2},i+1\}), (i + 1,
      \max\{h_{1},h_{2},i+1\}) && \text{if } i < k \\
      \hspace{-0.4cm}\label{trans:min-inform}
      \Trans{Inform}
      (i_1, h_1), (i_2, h_2) &\mapsto (i_1, \max\{h_1, h_2\}), (i_2, \max\{h_1, h_2\}) && \text{if }
                                                          i_{1} \ne i_{2}
    \end{align*}

    \smallskip Note that we no longer need to pull up agents to the topmost
    level. This functionality is completely replaced by the second
    state component.
  }
\end{protocol}

\begin{theorem}\label{lem:robust-min}
  \ref{prot:robust-min} computes \(\min(n,k)\) robustly.
\end{theorem}
\begin{proof}
 Let \(\pi =
  \Cfg_{0},\Cfg_{1}, \ldots\) be a fair execution with \(s\)
  snipes starting at the initial configuration of \ref{prot:robust-min} that puts $n$ agents in its input state. After the last snipe, all subsequent configurations have \(n-s\) agents.

  We first show that some agent eventually reaches level 
  \(\min(n{-}s,k)\) or above. As long as two agents share a level below \(k\),
  by fairness they eventually interact and one climbs via
  transition~\ref{trans:min-step} to the next level. So eventually each level
  below~\(k\) has at most one agent, and by the pigeonhole principle at
  least one agent is at level \(\min(n{-}s,k)\) or above. Further,
  all transitions preserve the invariant that
  if level \(i\) is occupied then so are all levels \(1, \ldots,
  i{-}1\). Hence no agent ever reaches a level above \(n\).

  Let \(h\) be the largest second component of the state of any agent after
  the last snipe.  Some agent must have been at level \(h\) at some
  point, and so \(h \le n\). Further, all transitions preserve the invariant that the second component of the state of an agent is at least as large as the first. So \(h \ge \min(n{-}s,k)\), which implies that \(h\) is a permissible output of the protocol.  By fairness,  transition~\ref{trans:min-inform} eventually propagates \(h\) to all agents.
\end{proof}

%% file: sec-mod.tex
We now consider functions \(f(n) = n \bmod m\)
for some constant \(m\ge2\). In~\cite{LossinCEGP25} Lossin \etal
describe a protocol that robustly decides whether \(f(n) \in A\) for
some fixed set \(A\). We construct a robust protocol computing \(f(n)\). We start with some intuition about the construction.

The first insight we need is that if at least $m-1$ agents are sniped,
then every output is permissible for $(n \bmod m)$, so we mainly need to worry about the
case of at most $m-2$ snipes. Since this number is independent of $n$, we can apply redundancy: loosely speaking, we \emph{replicate} a protocol for $(n \bmod m)$, and return the output produced by the largest number of replicas. If the number of replicas is large enough, then, even if some of them fail, the output of the surviving ones is still correct. 

Our protocol extends \ref{prot:robust-min}. An agent's level within the
tower determines the replica he currently works for. Each agent
carries one pouch for each level. The agents at level $i$ count (modulo
$m$) how many agents they have seen pass through level $i$: whenever
two agents in level $i$ meet, one passes the pebbles in its $i$-th pouch
to the other agent and climbs to level $i+1$. There it starts counting
anew by placing a single pebble in his $(i+1)$-th pouch. If no agents
are sniped then the count in the $i$-th level will converge to $(n - i
+ 1) \bmod m$, as all agents are counted, except for the $i-1$ agents
occupying the lower levels. Hence $n \bmod m$ can be determined by
simply adding $i - 1$ to this count. Sniping an agent in level $i$ may
have an arbitrary effect on the count in that replica; however, in other
replicas, at most the sniped agent itself is not counted. If
the number of replicas is large enough, this ensures that the most
frequently occurring output among the replicas is permissible.

Agents in different levels inform each other about the counts in their
respective replicas. Each then chooses the most frequently given
result as final output. To ensure convergence even if no agent remains
for a replica, if two agents disagree on the count for a replica
neither of which work for, they both reset it to 0. 
Finally, if not enough agents are available to populate all levels,
the output is determined from the highest occupied level.

\begin{protocol}{RobustMod}\label{prot:robust-mod}
  \States{Let $\Copies = m^{2}$.
    \par\noindent\hspace*{\labelwidth}$\begin{array}[t]{@{}r @{\;} l @{\quad} l}
      Q \coloneq & [\Copies{+}1] & \text{current level of the agent} \\
      \times & [\Copies{+}1] & \text{highest level the agent knows has been populated} \\
      \times & \Z_{m}^{\Copies} & \text{agent's information on the counts in each replica}
    \end{array}$
% The first two components of the state behave as in \ref{prot:robust-min}. An agent at level $1 \leq i \leq \Copies$ works for the $i$-th replica.
}
  \Input{$I \coloneq \set{(1,1,(1, 0, \ldots, 0))}$}
  \Output{Given $\vec u \in \Z_{m}^{\Copies}$,  let $\pl(\vec u)$, the \emph{plurality} of $\vec u$, be the value that  occurs most frequently in the components of $\vec u$, or the smallest one in case of a tie.\\
    \par\noindent\hspace*{\labelwidth}$
     O(i,t,\vec u) \coloneq \begin{cases}
        t \bmod m & \text{if } t \le \Copies \\
        \pl(\vec u + (0,1,2,\ldots c-1)) & \mbox{otherwise}
  \end{cases}$\\
   }
  \Transitions{ Intuitively,  an agent that moves up in the tower transfers its current count to the agent 
    remaining behind, and starts counting anew in the next level. Agents exchange information on the highest  level they know has been populated, and on the counts of all replicas. Agents are an authority on the output of the replica they currently work for, which in an interaction is just copied by others. If two agents
    disagree on the output of a replica neither of them is an authority on,
    they both reset their values for it to $0$. This ensures convergence in the
    case that an agent is sniped.  For the formal description, let $\vec u \reconcile \vec v$ denote the \emph{reconciliation} of two 
    vectors, defined by $(\vec u \reconcile \vec v)_j = u_j$ if $u_j = v_j$, and $0$
    otherwise. There are two sets of transitions, depending on whether the interacting agents occupy the same level, different from the top one, or not. The first one is:
    \begin{align*}
      \label{trans:mod-move-up}
      \Trans{MoveUp}
      \quad (i, t_1, \vec u), (i, t_2, \vec v) &\mapsto (i, t', \vec w), (i + 1, t', \vec w) &
                                                               \text{if
                                                               } i \le \Copies
    \end{align*}
    where $\vec w$ is obtained from $\vec u \reconcile \vec v$ by replacing the $i$-th
    component with $(u_i{+}v_i) \bmod m$, and, if $i < \Copies$, the
    $(i{+}1)$-th component with $1$; and $t' := \max\{t_1, t_2, i{+}1\}$. The second one is:
    %$w := (u \reconcile v)[i \coloneq (u_i{+}v_i) \bmod m,\; i{+}1 \coloneq 1]$ and $t' := \max\{t_1, t_2, i{+}1\}$, and
    \begin{align*}
      \label{trans:mod-exchange}
      \Trans{Exchange}
      (i_1, t_1, \vec u), (i_2, t_2, \vec v) &\mapsto (i_1, t', \vec w), (i_2, t', \vec w) & \text{if } i_{1} \ne i_{2} \text{ or } i_{1} = i_{2} = \Copies+1
    \end{align*}
    where $\vec w$ is obtained from $\vec u \reconcile \vec v$ by
    replacing, if $i_1 \le \Copies$, the $i_1$-th component with
    $u_{i_1}$, and, if $i_2 \le \Copies$, the $i_2$-th component with
    $v_{i_2}$; and $t' := \max\{t_1, t_2\}$.    
    % $w = (u \reconcile v)[i_1 \coloneq u_{i_1},\; i_2 \coloneq v_{i_2}]$ and $t' = \max\{t_1, t_2\}$.
  }
\end{protocol}

We prove correctness and robustness of  \ref{prot:robust-mod}.
\begin{definition}\label[definition]{def:lvl-val-out}
Let  $q=(i, t, \vec u)$ be a state of \ref{prot:robust-mod}. The \emph{level} of $q$ is $\ell(q):=i$. The \emph{value vector} of $q$ is $\vvalues(q) \coloneq \vec u$. The  \emph{output vector} of $q$ is $ \voutputs(q) := \vec u + (0, 1, \ldots, \Copies -1)$. The value vector $\VValues(\Cfg)$ and output vector
$\VOutputs(\Cfg)$ of a configuration $\Cfg$ are given by:
\begin{align*}
    \Values_i(\Cfg) &\coloneq \sum_{(i,t,\vec u) \in Q} \Cfg(i,t, \vec u) \cdot (u_i \bmod m) &
    \VOutputs(\Cfg) & \coloneq   \bigg(\VValues(\Cfg) + (0, 1, \ldots, \Copies -1)  \bigg) \bmod m
\end{align*}
%\begin{align*}
%    \VValues(\Cfg) & \coloneq   \bigg( \sum_{q \in Q} \Cfg(q) \cdot \values(q) \bigg) \bmod m   &
%    \VOutputs(\Cfg) & \coloneq   \bigg(\VValues(\Cfg) + (0, 1, \ldots, \Copies -1)  \bigg) \bmod m
%\end{align*}
%  \begin{align*}
%    \Values_i(\Cfg) &\coloneq \sum_{(i,t,\vec v) \in Q} \Cfg(i,t, \vec v) \cdot (v_i \bmod m) &
%    \Outputs_i(\Cfg) &\coloneq (\Values_{i}(\Cfg) + c -i) \bmod m
%  \end{align*}
\end{definition}
Loosely speaking $\Values_{i}(C)$ counts (modulo $m$) how many agents have reached level $i$ or higher in $C$; it does not count agents in levels below $i$. The output vector ``compensates'' for this.

Consider an execution $\Cfg_0, \Cfg_1, \ldots$, possibly with snipes, of \ref{prot:robust-mod}.
For the proof it is convenient to assign identities to agents, which we assume are natural numbers. Let 
$q_{j}^{t}$ be the state of the $j$-th agent in $\Cfg_t$. We consider sniped agents to move to a special state $\bot$.

\begin{lemma}\label[lemma]{lem:mod-invariant}
  For every $t$ and $1 \le i \le \Copies$, if until step $t$ no agent in level $i$ has been
  sniped, then $\Values_{i}(\Cfg_t) \equiv \left| \setc*{j}{\exists t'\le t .\, \ell\left(q_{j}^{t'}\right) = i}\right| \pmod{m}$.
\end{lemma}
\begin{proof}
 We prove the claim by induction on \(t\). The base case is 
 \(t=0\). At this point $\ell(q_{j}^{0}) = 1$ and \(\vvalues(q_{j}^{0}) =
 (1,0,\ldots,0)\) for all agents \(j\), so $\Values_{1}(C_{0}) = n \bmod m$ and
 $\Values_{i}(C_{0}) = 0$ for all $i > 1$. For the induction step, consider the step from $C_{t}$ to $C_{t+1}$.
 We have either $C_{t} \step C_{t+1}$ or $C_{t}\snipe C_{t+1}$.  
 In the first case, the step is caused by either \ref{trans:mod-exchange} or
 \ref{trans:mod-move-up}. Transition \ref{trans:mod-exchange}
  does not change the level of either agent, and for both involved
  agents $\values_{\ell(q_{j}^{t})}(q_{j}^{t}) =
  \values_{\ell(q_{j}^{t+1})}(q_{j}^{t+1})$. So $\VValues(C_{t}) =
  \VValues(C_{t+1})$. Transition \ref{trans:mod-move-up}
 moves one agent from level \(i\) to level
  \(i+1\). Accordingly, as this agent sets the \(i+1\)-th component of the value vector to $1$, 
  $\Values_{i+1}(C_{t+1})$ increases by one compared to
  $\Values_{i+1}(C_{t})$. Further, $\Values_{i}(C_{t+1})$ does not change, as the
  contribution of the agent moving up is added to the other agent.
  Finally, if $C_{t}\snipe C_{t+1}$ happens by sniping an agent in level $i$,
  then only $\Values_{i}(C_{t+1})$ can change compared to $\Values_{i}(C_{t})$,
  and we explicitly exclude this case in the statement of the lemma.
\end{proof}
% \begin{proof}
%   We prove the claim by induction on \(t\). The base case \(t = 1\)
%   holds trivially, since all agents start in level 1 and initially
%   have value 1 for replica $i$. For the induction step, consider
%   all possible execution steps. Transition \ref{trans:mod-exchange} does
%   not change any \(\Values_i\), nor does it change the level of any agent.
%   Transition \ref{trans:mod-move-up} moves one agent from level \(i\) to
%   level \(i+1\). $\Values_{i+1}$ also increases by 1, since the agent
%   moving to level \(i+1\) has value 1 for replica \(i+1\).
%   Since its value for replica \(i\) is added to
%   the other agent, $\Values_{i}$ does not change.
%   Finally, sniping an agent in level \(i\) only changes \(\Values_i\),
%   which we explicitly exclude in the lemma statement.
% \end{proof}
\begin{theorem}\label{thm:robust-mod}
 \ref{prot:robust-mod} computes $(n \bmod m)$ robustly.
\end{theorem}
\begin{proof}
  Consider a fair execution \(\pi=\Cfg_{0},\Cfg_{1},\ldots\) with $s$
  snipes. As in \ref{prot:robust-min}, the protocol eventually reaches a configuration where each
  level $i \le \Copies$ has at most one agent, which we call the \emph{leader} of level $i$, 
  and all remaining agents are in level $\Copies{+}1$. After that, only
  \ref{trans:mod-exchange} transitions occur, which propagate
  each leader's value to all agents.  For levels whose leader has
  been sniped, disagreements are resolved to $0$ by reconciliation,
  so all agents eventually agree on the same value vector. More
  specifically, let $C$ be the configuration reached after
  convergence. Then, for levels $i$ in which no snipe has occurred,
  $\values_{i}(q) = \Values_{i}(C)$, since $\Values_{i}(C)$ is
  completely determined by the leader of level $i$.

 This already shows that the protocol always converges to some
result. We now show that that result is permissible. Since we compute
modulo $m$, for $s \geq m-1$ any output is permissible by the definition of
robustness. Thus we only need to consider $s \leq m-2$ snipes.

  Let $s_{i}$ be the number of agents that have been sniped in level
  $i$. Then there exist $s_{i}' \le s_{i}$, such that the number of
  agents which have at some point been in level $i$ is \(n - i + 1 -
  \sum_{j=1}^{i-1}s_{j}'\).  This is because all agents that have
  not been sniped in a lower level, or are a leader of a lower level,
  have at some point passed through level $i$.
  Hence, by \Cref{lem:mod-invariant}, for any level $i$ with $s_{i} =
  0$, we have
  $\Values_{i}(C) \equiv n - i + 1 - \sum_{j = 1}^{i-1} s_{j}' \pmod{m}$ and
  therefore 
  $\Outputs_{i}(C) \equiv n - \sum_{j = 1}^{i-1} s_{j}' \pmod{m}$.
  Since there are $s$ snipes total, at most $s$ levels can have $s_i
  > 0$. Thus at least $\Copies - s$ levels remain unaffected.
  The output in each of these levels is permissible, as it is of the
  form $n - k \bmod m$ for some $k \le s$. Since we assumed $s \le m-2$ there are
  at most $m-1$ permissible outputs. 
  By the pigeonhole principle, one of these must occur in at least
  $\frac{\Copies - s}{m-1}$ replicas, which is more than the maximum
  number of replicas affected by snipes:
  \begin{equation*}
    \frac{\Copies - s}{m-1} \geq \frac{m^2 - m + 2}{m - 1} = m + \frac{2}{m-1} > m > s
  \end{equation*}

  Thus among the outputs of all replicas there is a plurality for a permissible output.
\end{proof}

%%% Local Variables:
%%% TeX-master: "robust-population-protocols"
%%% End:

%% file: sec-min-mod.tex
We show that $\min(n, k)$ and $(n \bmod m)$ can be
computed \emph{simultaneously} and robustly by a slight modification of 
\ref{prot:robust-mod}. Recall that the first component of the states
of \ref{prot:robust-mod} model a tower of height $\Copies$.
It suffices to extend the height to $\max(\Copies{+}1, k{+}m)$. The resulting protocol is:

\begin{protocol}{RobustMinMod}\label{prot:robust-min-mod}
  \States{ Let $\Copies \coloneq m^2$ and $\Height \coloneq \max(\Copies{+}1, k{+}m)$.
    \par\noindent\hspace*{\labelwidth}$\begin{array}[t]{@{}r @{\;} l @{\quad} l}
      Q \coloneq & [\Height] & \text{agent's own level} \\
      \times & [\Height] & \text{highest level the agent knows of} \\
      \times & \Z_{m}^{\Copies} & \text{values computed in each mod replica}
    \end{array}$}
 \Input{$I \coloneq \set{(1,1,(1, 0, \ldots, 0))}$}
  \Output{$O(i,t,\vec u) \coloneq \begin{cases}
       \big( \min(t,k),\; t \bmod m \big) & \text{if } t < \Height \\
       \big( \min(t,k),\; \pl(\vec u + (0,1,2,\ldots c-1)) \big) & \text{otherwise}
     \end{cases}$}
  \Transitions{As in \ref{prot:robust-mod}, with $\Copies$ replaced by $\Height$.}
\end{protocol}

 \Cref{lem:robust-min,thm:robust-mod} already show that the
 individual components of the output given by
 \ref{prot:robust-min-mod} are permissible under 
 robustness. However, we need to show that the combination of both is
 also permissible. Consider for example the case where $k = 5$, $m = 3$,
 $n = 6$. In a run with 2 snipes, 4 would be a permissible output for
 the first component ($\min(n,k)$) and 2 would be a permissible output
 for the second component ($n \bmod m$). However, achieving either
 requires removing different amounts of agents initially, and hence
 the combination of both is not permissible. We show that such a
 situation cannot occur.

We define the level \(\ell\), the value vector \(\values\) and output
vector \(\outputs\) of a state as well as the value vector \(\Values\)
and output vector \(\Outputs\) of a configuration the same as for the
\ref{prot:robust-mod} protocol in \Cref{def:lvl-val-out}.

\begin{theorem}\label{thm:robust-min-mod}
\ref{prot:robust-min-mod} 
 computes \((\min(n,k), n \bmod m)\) robustly. 
\end{theorem}
\begin{proof}
 \Cref{lem:robust-min,thm:robust-mod} already show that all the
 agents will come to a consensus on both the tower height $h$ and the
 value vector $\values$ for the mod-replicas. Hence it only
 remains to show that the combination is permissible. Let $\Cfg$ be the final
 configuration reached. We distinguish two cases:

\smallskip\noindent  \emph{Case 1. $h < \Height$.}
   In this case the two output components are $\min(h, k)$ and
   $h \bmod m$. These are obviously compatible.

\smallskip\noindent  \emph{Case 2. $h = \Height$.}
   In this case the output components are $\min(h, k) = k$ and $\pl(\voutputs_i)$.
   Let $s$ be the total number of snipes, and let $s_{i}$ be the
   number of snipes on level $i$. If \(s \ge m - 1\), then any mod-output is permissible. Since $n \ge \Height$ and $\Height - k \ge m$, for any $r \in \Z_m$ we can remove $(n-r) \bmod m$ agents to achieve mod-output $r$ while still having at least $k$ agents remaining, so $(k, r)$ is permissible.

   In the following we assume that \(s < m-1\). Recall from the proof of
   \Cref{thm:robust-mod}, that there exist $s_{i}' \le s_{i}$, such that
   the number of agents which have at some point been on level \(i\)
   is \(n - i+1 - \sum_{j=1}^{i-1} s_{j}'\), and then by
   \Cref{lem:mod-invariant}, for all levels \(i\) without snipes \(\Outputs_{i}(\Cfg) \equiv n - \sum_{j=1}^{i-1}
   s_{j}' \pmod m\).
   On the other hand, for $h = \Height$ to be
   possible, at least one agent must have reached level $\Height$, which means
   that at least $\Height - i + 1$ must have reached level $i$ at some point:
   $$\begin{array}{rcl}
     \big( n - i + 1 - \sum_{j=1}^{i-1}s_{j}' \ge \Height - i + 1 \big) 
     & \text{if{}f} & 
     \big(n - \sum_{j=1}^{i-1}s_{j}' \ge \Height \big)
   \end{array}$$
   For $\Outputs_{i}(\Cfg)$, where $s_{i} = 0$, this implies
     \(\Outputs_{i}(\Cfg) \equiv n' \pmod{m}\)
   for some $\Height \le n' \le n$.
   And then, since $k < \Height$, for each level $i$ without snipes, $(k,
   \Outputs_{i}(\Cfg))$ is a permissible combination of output
   components. Since \(s < m-1\), the plurality of the outputs
   corresponds to \(\Outputs_{i}\) for some level \(i\) without snipes by the
   same argument as in the proof of \Cref{thm:robust-mod}. Hence this
   protocol is also robust.
\end{proof}

%%% Local Variables:
%%% TeX-master: "robust-population-protocols"
%%% End:

%% file: sec-monadic-predicates.tex
Given functions $f^1 \colon \mathbb{N} \to K_1, \ldots, f^n \colon \mathbb{N} \to K_n$ of arity $1$ and robust protocols $P^1, \ldots, P^n$ computing $f^1, \ldots f^n$, we show how to construct a robust protocol $P$ computing the function $f \colon \mathbb{N}^n \to K_1\times \cdots \times K_n$ of arity $n$ given by 
$f(x_1, \ldots, x_n) := (f^1(x_1), \ldots, f^n(x_n))$.

\smallskip\noindent \textit{Remark.} Observe that $f$ is again a function $\mathbb{N}^n \to K$ for
$K:= K_1\times \cdots \times K_n$. In other words, if $f^1, \ldots, f^n$ are functions
with a finite range, then so is $f$. However, if $f^1, \ldots, f^n$ are
predicates, i.e., $K_1 = \cdots = K_n=\{0,1\}$, then $f$ is not a predicate. This explains why we defined robustness for functions, not just predicates.

%Observe that we compute the function $f(x_1, \ldots, x_n)$ of arity $n$ and \emph{not}  the function $f' \colon \mathbb{N} \to K_1\times \cdots \times K_n$ of arity 1 given by $f'(x) := (f^1(x), \ldots, f^n(x))$. Let us see why the latter is intuitively harder. 
%
%Since $f^i$ has arity $1$, the protocol $P_i$ computing $f^i$ has one single initial state, and outputs $f^i(k)$ for every initial configuration with $k$ agents in it and no agents in any other state.  So, intuitively, the protocol has $k$ agents at its disposal to compute $f(k)$. The protocol $P$ computing $f(x_1, \ldots, x_n)$ has $n$ initial states, and outputs $f(k_1, \ldots, k_n)$ for every initial configuration $(k_1, \ldots, k_n)$. So, intuitively, $P$ has $k:=k_1+ k_2 + \ldots + k_n$ agents at its disposal to compute $(f^1(k_1), \ldots, f^n(k_n))$, and can in principle naturally split the $k$ agents into groups of $k_1, \ldots, k_n$ agents, with the group of $k_i$ agents ``responsible'' of computing $f^i(k_i)$. These groups can exchange information about the results of their own computations, so that eventually all agents reach a consensus about the tuple $(f^1(k_1), \ldots, f^n(k_n))$.
%
%Such a construction is no longer possible for computing $f'(x) := (f^1(x), \ldots, f^n(x))$. A protocol $P'$ computing $f'(x)$ has one single initial state, and so it must compute all of $f^1(k), \ldots, f^n(k)$ with only $k$ agents. In particular, the task cannot be distributed among groups of agents. 

We present a construction that, given robust protocols $P_i$ for $f^i(x_i)$, delivers a robust protocol $P$ for $f(x_1, \ldots, x_n)$. Using the construction it is easy to prove the existence of robust
protocols for all monadic Presburger predicates.

A subtlety here is
that, although a population protocol has no defined output on the empty
configuration, any individual variable may be $0$ initially or be
reduced to $0$ by snipes; our construction gives permissible outputs in either case.

\begin{definition}
Let $f\colon F \rightarrow F'$ and $g\colon G \rightarrow G'$. The \emph{parallel composition} of $f$ and $g$ is the function 
  \(f \times g \colon F \times G \rightarrow F' \times G'\) with  \( (a,b) \mapsto (f(a), g(b))\)
\end{definition}

Let $\Prot_{f} = (Q_{f}, I_{f}, O_{f}, \delta_{f})$ and $\Prot_{g} = (Q_{g},
I_{g}, O_{g}, \delta_{g})$ be population protocols
computing functions $f\colon \mathbb{N}^{I_f} \to F'$ and $g\colon \mathbb{N}^{I_g} \to G'$
respectively. We define the
parallel composition of these two protocols:

\begin{protocol}{$\Prot_f \times \Prot_g$}\label{prot:parallel}
  \States{\(Q \coloneq (\{f\} \times Q_f \times G') \cup (\{g\} \times Q_g \times F')\)
      
    Each agent computes either $f$ or $g$ and stores the current output
    of the other function.
  }
  \Input{
    \(
      I \coloneq \set{f} \times I_f \times \set{g(\mathbf{0})} \cup \set{g} \times I_g \times \set{f(\mathbf{0})}
      \)

    where $\mathbf{0}$ denotes the zero input.
  }
  \Output{
      \(O(f, q_f, o_g) \coloneq (O_f(q_f), o_g)\qquad 
      O(g, q_g, o_f) \coloneq (o_f, O_g(q_g))\)
  }
  \Transitions{
    If $h \in \{f, g\}$ and $(q_1, q_2) \to (q_1', q_2') \in \delta_h$:
    \begin{align*}
      \label{trans:parallel-exec}
      \Trans{Execution}
      (h, q_1, o), (h, q_2, o) &\mapsto (h, q_1', o), (h, q_2', o)
    \end{align*}
    Agents computing the same function execute their protocol if they
    agree on the other function's output.
    If they disagree, they reset to the zero output:
    \begin{align*}
      \label{trans:parallel-reset}
      \Trans{Reset}
      \begin{aligned}
      (f, q_1, o_1), (f, q_2, o_2) &\mapsto (f, q_1', g(\mathbf{0})), (f, q_2', g(\mathbf{0})) \\
      (g, q_1, o_1), (g, q_2, o_2) &\mapsto (g, q_1', f(\mathbf{0})), (g, q_2', f(\mathbf{0}))
      \end{aligned}
    \end{align*}
    where $(q_1, q_2) \to (q_1', q_2') \in \delta_f$ or $\delta_g$ respectively, and $o_1 \neq o_2$.

    Agents computing different functions exchange information:
    \begin{align*}
      \label{trans:parallel-exchange}
      \Trans{Exchange}
      (f, q_f, \_{}), (g, q_g, \_) &\mapsto (f, q_f, O_g(q_g)), (g, q_g, O_f(q_f))
    \end{align*}
  }
\end{protocol}

\begin{proposition}\label{prop:parallel-robust}
  If $\Prot_f$ and $\Prot_g$ are robust, then $\Prot_f \times \Prot_g$ is also robust.
\end{proposition}
\begin{proof}
 Let the initial configuration be
 $\Cfg = \set{f} \times C_f + \set{g} \times C_g$ with
 $C_f \in \N^{I_f}$ and $C_g \in \N^{I_g}$.
 Consider a fair execution $\pi$ with $s$ snipes, of which $s_f$ affect agents
 computing $f$ and $s_g$ affect agents computing $g$.

 If $C_g = \emptymultiset$, then $\Prot_f \times \Prot_g$ merely simulates
the robust protocol $\Prot_f$ on $C_f$, which gives a permissible
output for $f$. Since the initial (and correct) value of $g(\mathbf{0})$
cannot change, the overall output is permissible. The symmetric case
$C_f = \emptymultiset$ is analogous. In the remainder of the proof we
assume $C_f, C_g \neq \emptymultiset$.

 Let $\pi_{f},\pi_{g}$ be the projections of $\pi$ onto the $\Prot_{f}$
 and $\Prot_{g}$ components. Observe that $\pi_{f},\pi_{g}$ are also fair. Let $D_{f}$ be a
 configuration of $\Prot_{f}$ occurring infinitely often in
 $\pi_{f}$, and let $E_{f}$ be reachable from $D_{f}$. Since there are only finitely
 many configurations of $\Prot_{g}$ reachable from $C_{g}$, there must
 be at least one configuration $D_{g}$ such that  $\set{f}\times D_{f} +
 \set{g} \times D_{g}$ occurs infinitely often in $\pi$. Since the
 configuration $\set{f}\times E_{f} + \set{g}\times D_{g}$ is reachable and we
 assumed $\pi$ to be fair, it must also occur infinitely often. Hence
 $E_{f}$ occurs infinitely often in $\pi_{f}$. The same argument holds
 for $\pi_{g}$.

 By robustness of $\Prot_f$ and $\Prot_g$, the agents within each group
 eventually reach a stable consensus on permissible outputs $b \in F'$
 and $d \in G'$, respectively. If at least one agent remains in both groups,
 transition~\ref{trans:parallel-exchange} propagates $b$ and $d$ to all
 agents, and a stable consensus is reached.

 If all agents computing $g$
 have been sniped, then the surviving
 \(f\)-agents either disagree on their stored \(g\)-output, or they all
 agree on some $d' \in G'$. In the first case,
 transition~\ref{trans:parallel-reset} resets them to $g(\mathbf{0})$,
 which is permissible since it corresponds to sniping all \(g\)-agents
 in $C_g$ before any interaction. In the second case, $d'$ was copied
 from a \(g\)-agent before it was sniped, so there exists a
 configuration $D_g \in \N^{Q_g}$ of $\Prot_g$ such that $C_g \steps D_g$
 and $D_g$ populates some state $q$ with $O_g(q) = d'$. Sniping all but this one
 agent in $D_g$ leaves the singleton $\multiset{q}$, which is a stable
 $d'$-consensus, so by robustness of $\Prot_g$ the output $d'$ is
 permissible for $\Abs{C_g}-1$ snipes, and hence also for $\Abs{C_g}$ snipes.
 The symmetric case where all agents computing $f$ have been sniped is analogous.
\end{proof}
\begin{theorem}
Robust population protocols exist for every monadic Presburger predicate.
\end{theorem}

\begin{proof}
 Let $\varphi \colon \N^{k} \rightarrow \set{0,1}$ be a monadic Presburger predicate.
Each atomic predicate occurring in $\varphi$ is either a threshold
predicate $\nu_{i,j}(x_{i}) \Leftrightarrow x_{i} \ge t_{i,j}$ for $t_{j} \in \N$, or a modulo predicate
$\mu_{i,j}(x_{i}) \Leftrightarrow (x_{i} \bmod m_{i,j}) \in M_{i,j} $ for some $m_{i,j} \in
\N$ and $M_{i,j} \subseteq \Z_{m_{i,j}}$. For every atomic predicate for
 variable $x_{i}$ define $t_{i} \coloneq \max_{j} t_{i,j}$ and $m_{i} \coloneq \lcm_{j} m_{i,j}$.
Then we have $\nu_{i,j}(x_{i})$ if{}f  $\nu_{i,j}(\min\set{x_{i}, t_{i}})$ and  
$\mu_{i,j}(x_{i})$ if{}f $\mu_{i,j}(x_{i} \bmod m_{i})$.
  By \Cref{thm:robust-min-mod}
  there exist robust protocols computing
  $f^{i}(x_{i})\coloneq(\min(x_{i},t_{i}), x_{i}\bmod m_{i})$ for every
  variable $x_{i}$ and by repeated application of
  \Cref{prop:parallel-robust} we can merge these into a single robust protocol computing
  $f(x_{1},\ldots,x_{k}) \coloneq (f^{1}(x_{1}),\ldots, f^{k}(x_{k}))$.
    Since all atomic predicates of $\varphi$ are determined by the output of $f$,
    adjusting the output function yields
    a robust protocol deciding $\varphi$.
\end{proof}

%%% Local Variables:
%%% TeX-master: "robust-population-protocols"
%%% End:

%% file: sec-lower-bound.tex
\section{State Complexity of Robust Computation}
\label{sec:lowerbound}

Consider the family $\{ x \geq k\}_{k \geq 0}$ of predicates. The family can be decided by protocols with $O(\log k)$ states \cite{BlondinEJ18,CzernerGHE24}. Further, an infinite subfamily can be decided by protocols with only $\Theta(\log\log k)$ states, and no infinite subfamily can be decided by protocols with $O(\log\log\log k)$ states  \cite{BlondinEJ18,Czerner24,CzernerEL23,CzernerGHE24,Leroux22}.  
However, a simple inspection shows that none of these protocols are robust\footnote{We do not further elaborate on this, because it is also a consequence of the results of this section.}. This raises the question of how succinctly (i.e., with the least possible number of states) can a family of protocols robustly decide $\{x \geq k\}_{k \geq 0}$. The \ref{prot:tower} protocols of Example \ref{ex:tower} decide $x \geq k$ with $k$ states.  We prove that they are optimal: every protocol robustly deciding $x \geq k$ has at least $k$ states.  We derive this result from a more general theorem giving a lower bound on the number of states of protocols that robustly decide \emph{asymptotically constant} predicates, defined next.

\begin{definition}
Let $\varphi : \mathbb{N}^I \to \{0,1\}$ be a predicate. An input
$\ICfg \in \mathbb{N}^I$ of $\varphi$ is \emph{upward invariant} if $\varphi(\ICfg') =
\varphi(\ICfg)$ for every $\ICfg' \in \mathbb{N}^I$ such that $\ICfg \leq \ICfg'$. The set of upward-invariant inputs of $\varphi$ is denoted $\mathcal{U}_\varphi$. 
\end{definition}

Upward-invariant inputs have the following property:
\begin{lemma}
\label{lem:upwardinvariant}
All upward-invariant inputs of a predicate have the same output.
\end{lemma}
\begin{proof}
Let $\varphi : \mathbb{N}^I \to \{0,1\}$ be a predicate and let $\ICfg_1, \ICfg_2 \in
\mathcal{U}_\varphi$.  Let $\ICfg \in \mathbb{N}^I$ be any input such that $\ICfg \geq
\ICfg_1$ and $\ICfg\geq \ICfg_2$.  Since $\ICfg_1$ and $\ICfg_2$ are
upward invariant, we have $\varphi(\ICfg) = \varphi(\ICfg_1)$ and $\varphi(\ICfg) =
\varphi(\ICfg_2)$, and so $\varphi(\ICfg_1) = \varphi(\ICfg_2)$.
\end{proof}

\begin{definition}
\label{def:ac}
A predicate $\varphi$ is \emph{asymptotically constant}, or an \emph{ac-predicate}, if $\mathcal{U}_\varphi$ is nonempty. An ac-predicate $\varphi$ is 
\begin{itemize}
\item \emph{asymptotically zero}, or an a0-predicate, if $\varphi(A) = 0$ for every $A \in  \mathcal{U}_\varphi$; and 
\item \emph{asymptotically one}, or an a1-predicate, if $\varphi(A) = 1$ for every $A \in  \mathcal{U}_\varphi$. 
\end{itemize}
\end{definition}
%Observe that a predicate is asymptotically constant if{}f there exists $k \in \mathbb{N}$ such that $\varphi(\ICfg)$ has the same value for every $\ICfg$ satisfying $\ICfg(x) \geq k$ for every $x \in I$. 
\begin{example}
All predicates $x \geq k$ are a1-predicates with $\mathcal{U}_{x \geq k} = \{k' \in \mathbb{N} \mid k' \geq k\}$. The predicate $x+y \leq 4$ is an example of an a0-predicate. Predicates like $x-y\geq 0$ or $x \bmod m \geq k$ are not asymptotically constant.  %It is easy to see that ac-predicates are closed under Boolean operations.  
\end{example}

\noindent The case distinction of Definition \ref{def:ac} is exhaustive by Lemma \ref{lem:upwardinvariant}. Our goal is to prove:
\begin{theorem}
\label{thm:lowerbound}
Every robust population protocol computing an ac-predicate
$\varphi$ has at least $\min_{\ICfg \in \mathcal{U}_\varphi} \max_{x \in I} \ICfg(x)$
states.
\end{theorem}
\noindent This bound implies:
\begin{corollary}
\label{cor:lowerbound}
Every robust protocol deciding $x \geq k$ has at least $k$ states. In particular, the \ref{prot:tower} protocol deciding $x \geq k$ has a minimal number of states. 
\end{corollary}
\begin{proof}
Since $x \geq k$ is an ac-predicate, Theorem \ref{thm:lowerbound} applies.
Let $\Prot=(Q,I,O,\delta)$ be an arbitrary  protocol deciding $x \geq k$. $\Prot$ has one single initial state, which we can also call $x$. So $I=\{x\}$ and the bound of Theorem \ref{thm:lowerbound} becomes $\min_{\ICfg \in \mathcal{U}_{x \geq k}} \ICfg(x)$. Since 
$\mathcal{U}_{x \geq k} = \{k' \in \mathbb{N} \mid k' \geq k\}$, we get $|Q| \geq \min \{k' \in \mathbb{N} \mid k' \geq k\} =  k$.
\end{proof}

\subsection{Proof of Theorem \ref{thm:lowerbound}.}  We only prove Theorem \ref{thm:lowerbound} for a1-predicates, the proof for a0-predicates being completely analogous. We proceed in two steps.  First, we introduce the \emph{critical input} of a protocol deciding a given a1-predicate, and show that if this input is upward invariant, then the protocol has at least the number of states of Theorem \ref{thm:lowerbound} (Proposition \ref{prop:reduction}). Then we prove that the critical input is indeed upward invariant (Proposition \ref{prop:critical}). 

\smallskip
\noindent\textbf{Notation}: We fix an a1-predicate $\varphi$ and a robust protocol $R=(Q, I, O, \delta)$ deciding $\varphi$.  We let  $Q_0$ and $Q_1$ denote the sets of states $q \in Q$ with output 0 and 1, and call them the sets of \emph{rejecting} and \emph{accepting} states of $R$. Observe that $Q = Q_0 \cup Q_1$.

\subparagraph{The critical input.} Loosely speaking, the critical input of $R$ puts in each initial state as many agents as the number of rejecting states of $R$ plus 1.

\begin{definition}
The \emph{critical input} $\ICfg_R \in \mathbb{N}^I$ of $R$ is the input $\ICfg_R$ given by $\ICfg_R(x) := |Q_0| + 1$ for all $x \in I$ (recall that $Q_0$ and $Q_1$ are the rejecting and accepting states of $R$).
\end{definition}

\begin{proposition}
\label{prop:reduction}
If $\ICfg_R$ is upward invariant, then $R$ has at least $\displaystyle \min_{\ICfg \in \mathcal{U}_{\varphi}} \max_{x \in I} \ICfg(x)$ states.
\end{proposition}
\begin{proof}
Since $R$ decides an a1-predicate, it has at least one accepting state, and so $|Q| \geq |Q_0| + 1$.  By definition, $|Q_0|+ 1=\ICfg_R(x)$ for every $x \in I$, and so $|Q_0| + 1 = \max_{x \in I} \ICfg_R(x)$. Since  $\ICfg_R$ is upward invariant, we have $\ICfg_R \in \mathcal{U}_{\varphi}$ and so $\max_{x \in I} \ICfg_R(x) \geq  \min_{\ICfg \in \mathcal{U}_{\varphi}} \max_{x \in I} \ICfg(x)$.
\end{proof}

\subparagraph{The critical input is upward invariant.}
We first characterize the set of inputs of $R$ with output 1 (Lemma \ref{lem:charout1}).

\begin{definition}
Let $Q' \subseteq Q$ be a set of states and let $\Cfg$ be a configuration.  
\begin{itemize}
\item $\Cfg$ is \emph{confined to $Q'$} if $\supp{\Config{D}} \subseteq Q'$ for every configuration $\Config{D}$ reachable from $\Cfg$.
\item $\Cfg$ is \emph{confinable to $Q'$} if some configuration reachable from $\Cfg$ is confined to $Q'$. 
\end{itemize}
\end{definition}

\begin{lemma}
\label{lem:charout1}
An input $\ICfg$ of $R$ has output 1 if{}f $\ICfg$ is \emph{not} confinable to $Q_0$.
\end{lemma} 
\begin{proof}
Since $R$ decides $\varphi$, for every input $\ICfg$ of $R$ and every $b \in \{0,1\}$ we have: (1) if $\ICfg$ is confinable to $Q_b$, then $\varphi(\ICfg) = b$, and (2)  $\ICfg$ is either confinable to $Q_0$ or confinable to $Q_1$. The result follows.
\end{proof}

Lemma \ref{lem:charout1} allows us to redefine our goal: instead of proving that the critical input is upward invariant, i.e., that $\ICfg \geq \ICfg_R$ implies $\varphi(\ICfg)=\varphi(\ICfg_R)=1$, we equivalently prove that $\ICfg \geq \ICfg_R$ implies that $\ICfg$ is not confinable to $Q_0$. We do this in Proposition \ref{prop:critical}. First, we obtain  a sufficient condition for ``non-confinedness'' in Lemma \ref{lem:notconfinable}. 

\begin{definition}
Let $Q' \subseteq Q$ be a set of states. A transition $t$ of $R$ is an \emph{escape transition} of $Q'$ if
$\supp{\mathrm{pre}(t)} \subseteq Q'$ and $\supp{\mathrm{post}(t)} \nsubseteq Q'$.
\end{definition}

\begin{lemma}
\label{lem:notconfinable}
Let $Q' \subseteq Q$ and let $n := |Q'|$. Let $\Iconf{D}$ be a configuration such that (1) $|\Iconf{D}| > n$, and 
(2) for every configuration $\Conf{C}$ satisfying $\Iconf{D} \snipereach{n} \Conf{C}$, every set of states $S$ with $\supp{C} \subseteq S \subseteq Q'$ has an escape transition. Then $\Iconf{D}$ is not confined to $Q'$.
\end{lemma}
\begin{proof}
Given $k \geq 0$, we say that $\Iconf{D}$ is \emph{$k$-snipe-confinable to $Q'$} if some configuration $\Conf{C}$ such that $\Iconf{D} \snipereach{k} \Conf{C}$ is confined to $Q'$. 
Observe that $\Iconf{D}$ is confined to $Q'$ if{}f $\Iconf{D}$ is $0$-snipe-confinable to $Q'$. Let $\Iconf{D}$ be a configuration satisfying the conditions of the lemma.
We claim: 

\smallskip\parbox{0.9\textwidth}{
For all $i \leq n$ and all $S \subseteq Q'$ of size at most $i$: $\Iconf{D}$ is not $(n-i)$-snipe-confinable to $S$. }

\smallskip\noindent Once the claim is proved, the lemma follows by taking $i:=n$. Indeed, in this case $Q'$ itself has size at most $i$, and so we can set $S:=Q'$. This yields that $\Iconf{D}$ is not $0$-snipe-confinable to $Q'$, and so that $\Iconf{D}$ is not confined to $S$. We prove the claim by induction on $i$:

\smallskip\noindent\textbf{Base $i=0$.} In this case $S=\emptyset$. Since $|\Iconf{D}| > n$ by assumption (1), every configuration $\Conf{C}$ such that $\Iconf{D} \snipereach{n} \Conf{C}$ satisfies $|\Conf{C}| > 0$, and so $C$ is not confined to the empty set.

\smallskip\noindent\textbf{Step $ i \to i+1$.}
Let $S$ be an arbitrary set of states of size at most $i+1$, and let $\Conf{C}$ be a configuration with $\Iconf{D} \snipereach{n - (i+1)} \Conf{C}$. We prove that $\Conf{C}$ is not confined to $S$, i.e., that some configuration $\Conf{B}$ reachable from $\Conf{C}$ satisfies $\Conf{B}(q)>0$ for some $q \notin S$. From now on, we say that $\Conf{B}$ \emph{populates $q$}, and that $\Conf{C}$ \emph{can populate $q$}; in particular, $\Conf{C}$ is not confined to $S$ if{}f it can populate some state outside $S$.

If $\supp{C} \nsubseteq S$, then $\Conf{C}$ already populates a state outside $S$.
Otherwise $\supp{C} \subseteq S \subseteq Q'$, and so by our assumption $S$ has an escape transition $t$.
There are two cases: $|\supp{\mathrm{pre}(t)}| = 2$ and $|\supp{\mathrm{pre}(t)}| = 1$. We only consider the first one, the second being analogous. If $|\supp{\mathrm{pre}(t)}| = 2$, then $\mathrm{pre}(t) = \multiset{q_1, q_2 }$ with $q_1 \neq q_2$. There are three possible cases:
\begin{enumerate}
\item If $C(q_1) = C(q_2) = 0$, define $S' := S \setminus \{q_1\}$ and let $\Conf{C}'$ be any configuration 
obtained by sniping one agent from $\Conf{C}$. Then $\Iconf{D} \snipereach{n-(i+1)} \Conf{C} \snipereach{1} \Conf{C}'$ and so $\Iconf{D} \snipereach{n-i} \Conf{C}'$. By induction hypothesis applied to $\Conf{C}'$ and $S'$, the configuration $\Conf{C}'$
can populate a state in $Q \setminus S'$. Since $\Conf{C} \geq \Conf{C}'$, this state can also be populated from
$\Conf{C}$. If the state is $q_1$, we proceed with (2) or (3); otherwise we have already populated
a state outside $S$.
\item If $\Conf{C}$ populates either $q_1$ or $q_2$ (w.l.o.g.\ $q_1$), define
$S' := S \setminus \{q_2\}$ and let $\Conf{C}' := \Conf{C} - \multiset{q_1}$. As in the previous case we have 
$\Iconf{D} \snipereach{n-i} \Conf{C}'$, and so by induction hypothesis applied to $\Conf{C}'$ and $S'$
some configuration $\Conf{C}''$ reachable from $\Conf{C}'$ populates some state outside $S'$, which is either $q_2$ or some state outside $S$. It follows $\Conf{C} = \Conf{C}' + \multiset{q_1} \to^\star \Conf{C}'' + \multiset{q_1}$,  and so $\Conf{C}$ can either populate a state outside $S$, in which case we are done, or populate $q_2$ in addition to the agent already in \(q_1\), in which case we proceed with (3).

%\noindent\smallskip\textbf{Case $|\supp{\mathrm{pre}(t)}| = 1$.} Then $\mathrm{pre}(t) = \multiset{q_1,q_1}$ for some state $q_1$. Define $S' := S \setminus \{q_1\}$. Analogous to the previous case.

%\begin{enumerate}
%\item If $C(q_1) = 0$, consider a configuration $\Conf{C}'$ obtained by sniping one agent from
%$\Conf{C}$. Then $\Iconf{D} \snipereach{n-(i+1)} \Conf{C} \snipe \Conf{C}'$ and so
%$\Iconf{D} \snipereach{n-i} \Conf{C}'$. By induction hypothesis applied to $\Conf{C}'$ and $S'$,
%the configuration $\Conf{C}'$ can populate a state in $Q \setminus S'$. Since $\Conf{C} \geq \Conf{C}'$, that state can also be populated from $\Conf{C}$. If the state is $q_1$, we proceed with (2) or (3); otherwise, we have already populated a state outside $S$.
%
%\item If $C(q_1) = 1$, consider $\Conf{C}' := \Conf{C} - \multiset{q_1, q_1}$. Again
%$\Iconf{D} \snipereach{n-i} \Conf{C}'$, so we apply the induction hypothesis to $\Conf{C}'$ and $S'$ to
%obtain a configuration $\Conf{C}''$ that populates a state outside $S'$. By monotonicity,
%$
%\Conf{C} = \Conf{C}' + \multiset{q_1, q_1} \to \Conf{C}'' +\multiset{q_1, q_1}$, and so $\Conf{C}$ can either populate a state outside $S$ or put at least two agents in $q_1$; in this case we 
%proceed with (3).
%
%\item If $C(q_1) \geq 2$, then $t$ is enabled and firing $t$ leads to a configuration populating a state outside $S$.
%\end{enumerate}

\item If $\Conf{C}$ populates both $q_1$ and $q_2$ then the escape transition $t$ is enabled, and firing
$t$ leads to a configuration populating a state outside $S$.
\end{enumerate}
\end{proof}

\begin{proposition}
\label{prop:critical}
The critical input $\ICfg_R$ of $R$ is upward invariant.
\end{proposition}
\begin{proof}
Let $\ICfg'_R \geq \ICfg_R$ arbitrary. We prove $\varphi(\ICfg_R')=\varphi(\ICfg_R)=1$, which shows that $\ICfg_R$ is upward invariant. 

\smallskip\noindent \textbf{Claim I}:  Every configuration \(\Conf{D}\) reachable from $\ICfg'_R$ (without snipes, i.e., \(\ICfg'_R \steps \Config{D}\)) satisfies premises (1) and (2) of Lemma \ref{lem:notconfinable} for $Q':=Q_0$. \\
 The first premise is $|\Conf{D}| > |Q_0|$. It follows from $|\ICfg_R'| \geq |\ICfg_R|> |Q_0|$, which holds by the definition of the critical input, and $|\ICfg_R'| = |\Conf{D}|$, which holds because \(\ICfg_R' \steps \Config{D}\). For the second premise, let $\Conf{C}$ be any configuration satisfying $\Conf{D} \snipereach{n} \Conf{C}$, and let $S$ be any set of states such that $\supp{C} \subseteq S \subseteq Q_0$. We prove that $S$ has an escape transition. 
Let $\ICfg$ be any upward invariant input of $\varphi$ (which exists because $\varphi$ is an ac-predicate by assumption) and let $m := \max_{x \in I} \ICfg(x)$. Since $\ICfg'_R \to^\ast \Conf{D} \snipereach{n} \Conf{C}$, we have $m \cdot \ICfg'_R \to^{\ast} m \cdot \Conf{D} \snipereach{mn} m \cdot \Conf{C}$. 

\smallskip\noindent \textbf{Claim II} (subclaim of claim I): The configuration $m \cdot \Conf{C}$ stabilizes to $1$. \\
For the proof, remember that \(n:=|Q_0|\). We have $(m \cdot \ICfg'_R)(x) \geq m (|Q_0|+1)=mn+m$  for all $x \in I$. Therefore, any configuration $E$ obtained by removing $mn$ agents from $m \cdot \ICfg'_R$ has at least $m$ agents in each input state and is therefore greater than or equal to $\ICfg$. Since $\ICfg$ is upward invariant, we have $\varphi(E)=1$. By robustness, any configuration  reached from $m \cdot \ICfg'_R$ with  at most \(mn\) snipes  stabilizes to $1$. In particular, since $m \cdot \ICfg'_R \to^{\ast} m \cdot \Conf{D} \snipereach{mn} m \cdot \Conf{C}$, the configuration $m \cdot \Conf{C}$ stabilizes to $1$, and claim II is proved.

\smallskip \noindent Since $\supp{ m \cdot \Conf{C} } = \supp{C}$, we have $\supp{m \cdot \Conf{C}} \subseteq S \subseteq Q_0$, i.e., $m \cdot \Conf{C}$ only populates states of $Q_0$. But then $S$ must have an escape transition: Otherwise $m \cdot \Conf{C}$ would be confined to $S \subseteq Q_0$, which would imply that \(m \cdot \Conf{C}\) stabilizes to \(0\), contradicting claim II. Hence claim I is proved.

\smallskip\noindent Since $\Conf{D}$ satisfies the two premises of Lemma \ref{lem:notconfinable} for $Q':=Q_0$, we can apply the lemma, yielding that $\Conf{D}$ is not confined to $Q_0$. Since \(\Conf{D}\) is reachable from \(\ICfg'_R\), it follows that \(\ICfg'_R\) is not confinable to \(Q_0\), and so that it cannot stabilize to output $0$. Since $R$ decides a predicate, $\ICfg_R'$ must therefore stabilize to $1$, and so $\varphi(\ICfg_R')=1$, and we are done.
\end{proof}

\begin{proof}[Proof of Theorem \ref{thm:lowerbound}]
Let $R$ be a robust protocol deciding $\varphi$.
By Lemma \ref{lem:upwardinvariant}, all upward-invariant inputs of $\varphi$ have the same output. W.l.o.g, let this output be $1$. By Proposition \ref{prop:critical},  the critical input of $R$ is upward invariant, and the result follows by Proposition \ref{prop:reduction}.
\end{proof}

%% file: sec-conclusion.tex
\section{Conclusions}
\label{sec:conclusions}
We have proved that every monadic Presburger predicate can be computed by a robust population protocol exhibiting optimal fault-tolerance against crash failures. Further, we have shown that achieving robustness has at least double exponential cost in terms of state complexity. Finally, we have given optimal robust protocols for the predicates $x \geq k$. In future work  we expect to decide whether all Presburger predicates admit a robust protocol, and determine the exact cost of robustness.

%% file: robust-population-protocols.bib
@inproceedings{HagueLRW20,
  author       = {Matthew Hague and
                  Anthony W. Lin and
                  Philipp R{\"{u}}mmer and
                  Zhilin Wu},
  title        = {Monadic Decomposition in Integer Linear Arithmetic},
  booktitle    = {{IJCAR} {(1)}},
  series       = {Lecture Notes in Computer Science},
  pages        = {122--140},
  publisher    = {Springer},
  year         = {2020}
}

@article{LunaFIISV20,
  author       = {Giuseppe Antonio Di Luna and
                  Paola Flocchini and
                  Taisuke Izumi and
                  Tomoko Izumi and
                  Nicola Santoro and
                  Giovanni Viglietta},
  title        = {Fault-tolerant simulation of population protocols},
  journal      = {Distributed Comput.},
  volume       = {33},
  number       = {6},
  pages        = {561--578},
  year         = {2020}
}

@article{LunaFIISV19,
  author       = {Giuseppe Antonio Di Luna and
                  Paola Flocchini and
                  Taisuke Izumi and
                  Tomoko Izumi and
                  Nicola Santoro and
                  Giovanni Viglietta},
  title        = {Population protocols with faulty interactions: The impact of a leader},
  journal      = {Theor. Comput. Sci.},
  volume       = {754},
  pages        = {35--49},
  year         = {2019}
}

@inproceedings{Leroux22,
  author       = {J{\'{e}}r{\^{o}}me Leroux},
  title        = {State Complexity of Protocols with Leaders},
  booktitle    = {{PODC}},
  pages        = {257--264},
  publisher    = {{ACM}},
  year         = {2022}
}

@article{CzernerEL23,
  author       = {Philipp Czerner and
                  Javier Esparza and
                  J{\'{e}}r{\^{o}}me Leroux},
  title        = {Lower bounds on the state complexity of population protocols},
  journal      = {Distributed Comput.},
  volume       = {36},
  number       = {3},
  pages        = {209--218},
  year         = {2023}
}

@inproceedings{BlondinEJ18,
  author       = {Michael Blondin and
                  Javier Esparza and
                  Stefan Jaax},
  title        = {Large Flocks of Small Birds: on the Minimal Size of Population Protocols},
  booktitle    = {{STACS}},
  series       = {LIPIcs},
  volume       = {96},
  pages        = {16:1--16:14},
  publisher    = {Schloss Dagstuhl - Leibniz-Zentrum f{\"{u}}r Informatik},
  year         = {2018}
}

@article{Haase18,
  author       = {Christoph Haase},
  title        = {A survival guide to {P}resburger arithmetic},
  journal      = {{ACM} {SIGLOG} News},
  volume       = {5},
  number       = {3},
  pages        = {67--82},
  year         = {2018}
}

@article{ElsasserR18,
  author       = {Robert Els{\"{a}}sser and
                  Tomasz Radzik},
  title        = {Recent Results in Population Protocols for Exact Majority and Leader
                  Election},
  journal      = {Bull. {EATCS}},
  volume       = {126},
  year         = {2018}
}

@article{AlistarhG18,
  author       = {Dan Alistarh and
                  Rati Gelashvili},
  title        = {Recent Algorithmic Advances in Population Protocols},
  journal      = {{SIGACT} News},
  volume       = {49},
  number       = {3},
  pages        = {63--73},
  year         = {2018}
}

@incollection{AspnesR09,
  author       = {James Aspnes and
                  Eric Ruppert},
  title        = {An Introduction to Population Protocols},
  booktitle    = {Middleware for Network Eccentric and Mobile Applications},
  pages        = {97--120},
  publisher    = {Springer},
  year         = {2009}
}

@inproceedings{LossinCEGP25,
  author       = {Benno Lossin and
                  Philipp Czerner and
                  Javier Esparza and
                  Roland Guttenberg and
                  Tobias Prehn},
  title        = {The Black Ninjas and the Sniper: On Robust Population Protocols},
  booktitle    = {Principles of Verification {(3)}},
  series       = {Lecture Notes in Computer Science},
  volume       = {15262},
  pages        = {206--233},
  publisher    = {Springer},
  year         = {2024}
}

@article{AngluinADFP06,
  author       = {Dana Angluin and
                  James Aspnes and
                  Zo{\"{e}} Diamadi and
                  Michael J. Fischer and
                  Ren{\'{e}} Peralta},
  title        = {Computation in networks of passively mobile finite-state sensors},
  journal      = {Distributed Comput.},
  volume       = {18},
  number       = {4},
  pages        = {235--253},
  year         = {2006},
  url          = {https://doi.org/10.1007/s00446-005-0138-3},
  doi          = {10.1007/S00446-005-0138-3},
  bibsource    = {dblp computer science bibliography, https://dblp.org}
}

@article{CzernerGHE24,
  author       = {Philipp Czerner and
                  Roland Guttenberg and
                  Martin Helfrich and
                  Javier Esparza},
  title        = {Fast and succinct population protocols for Presburger arithmetic},
  journal      = {J. Comput. Syst. Sci.},
  volume       = {140},
  pages        = {103481},
  year         = {2024}
}

@article{AngluinAER07,
  author       = {Dana Angluin and
                  James Aspnes and
                  David Eisenstat and
                  Eric Ruppert},
  title        = {The computational power of population protocols},
  journal      = {Distributed Comput.},
  volume       = {20},
  number       = {4},
  pages        = {279--304},
  year         = {2007}
}

@inproceedings{Delporte-GalletFGR06,
  author       = {Carole Delporte{-}Gallet and
                  Hugues Fauconnier and
                  Rachid Guerraoui and
                  Eric Ruppert},
  editor       = {Phillip B. Gibbons and
                  Tarek F. Abdelzaher and
                  James Aspnes and
                  Ramesh R. Rao},
  title        = {When Birds Die: Making Population Protocols Fault-Tolerant},
  booktitle    = {Distributed Computing in Sensor Systems, Second {IEEE} International
                  Conference, {DCOSS} 2006, San Francisco, CA, USA, June 18-20, 2006,
                  Proceedings},
  series       = {Lecture Notes in Computer Science},
  volume       = {4026},
  pages        = {51--66},
  publisher    = {Springer},
  year         = {2006},
  url          = {https://doi.org/10.1007/11776178\_4},
  doi          = {10.1007/11776178\_4},
  bibsource    = {dblp computer science bibliography, https://dblp.org}
}

@inproceedings{GuerraouiR09,
  author       = {Rachid Guerraoui and
                  Eric Ruppert},
  editor       = {Susanne Albers and
                  Alberto Marchetti{-}Spaccamela and
                  Yossi Matias and
                  Sotiris E. Nikoletseas and
                  Wolfgang Thomas},
  title        = {Names Trump Malice: Tiny Mobile Agents Can Tolerate Byzantine Failures},
  booktitle    = {Automata, Languages and Programming, 36th International Colloquium,
                  {ICALP} 2009, Rhodes, Greece, July 5-12, 2009, Proceedings, Part {II}},
  series       = {Lecture Notes in Computer Science},
  volume       = {5556},
  pages        = {484--495},
  publisher    = {Springer},
  year         = {2009},
  url          = {https://doi.org/10.1007/978-3-642-02930-1\_40},
  doi          = {10.1007/978-3-642-02930-1\_40},
  bibsource    = {dblp computer science bibliography, https://dblp.org}
}

@article{AlistarhDKSU17,
  author       = {Dan Alistarh and
                  Bartlomiej Dudek and
                  Adrian Kosowski and
                  David Soloveichik and
                  Przemyslaw Uznanski},
  title        = {Robust Detection in Leak-Prone Population Protocols},
  journal      = {CoRR},
  volume       = {abs/1706.09937},
  year         = {2017},
  url          = {http://arxiv.org/abs/1706.09937},
  eprinttype    = {arXiv},
  eprint       = {1706.09937},
  bibsource    = {dblp computer science bibliography, https://dblp.org}
}

@inproceedings{Alistarh0U21,
  author       = {Dan Alistarh and
                  Martin T{\"{o}}pfer and
                  Przemyslaw Uznanski},
  editor       = {Avery Miller and
                  Keren Censor{-}Hillel and
                  Janne H. Korhonen},
  title        = {Comparison Dynamics in Population Protocols},
  booktitle    = {{PODC} '21: {ACM} Symposium on Principles of Distributed Computing,
                  Virtual Event, Italy, July 26-30, 2021},
  pages        = {55--65},
  publisher    = {{ACM}},
  year         = {2021},
  url          = {https://doi.org/10.1145/3465084.3467915},
  doi          = {10.1145/3465084.3467915},
  bibsource    = {dblp computer science bibliography, https://dblp.org}
}

@inproceedings{Czerner24,
  author       = {Philipp Czerner},
  editor       = {Dan Alistarh},
  title        = {Breaking Through the {\(\Omega\)}(n)-Space Barrier: Population Protocols
                  Decide Double-Exponential Thresholds},
  booktitle    = {38th International Symposium on Distributed Computing, {DISC} 2024,
                  Madrid, Spain, October 28 - November 1, 2024},
  series       = {LIPIcs},
  volume       = {319},
  pages        = {17:1--17:18},
  publisher    = {Schloss Dagstuhl - Leibniz-Zentrum f{\"{u}}r Informatik},
  year         = {2024},
  url          = {https://doi.org/10.4230/LIPIcs.DISC.2024.17},
  doi          = {10.4230/LIPICS.DISC.2024.17},
  bibsource    = {dblp computer science bibliography, https://dblp.org}
}
